\documentclass[a4paper,11pt]{article}
\usepackage[utf8x]{inputenc}

\usepackage{pstricks,pst-node,pst-tree}
\usepackage{newicktree}

\usepackage[arrow, matrix, curve]{xy}
\usepackage{array,amsmath,amssymb,amsthm,verbatim,epsfig}
\usepackage{amsfonts}
\usepackage{graphicx}
\usepackage{setspace}

\theoremstyle{plain}
\newtheorem{theorem}{Theorem}

\newtheorem{algorithm}[theorem]{Algorithm}

\newtheorem{corollary}[theorem]{Corollary}

\newtheorem{definition}[theorem]{Definition}

\newtheorem{lemma}[theorem]{Lemma}

\newtheorem{remark}[theorem]{Remark}

\numberwithin{equation}{section}

\newcommand{\suchThat}{\lvert}

\title{Estimating Species Trees from Quartet Gene
Tree Distributions under the Coalescent Model} \author{Martin Kreidl}

\begin{document}

\onehalfspacing
\maketitle

\begin{abstract} In this article we propose a new method,
	which we name `quartet neighbor joining', or
	`quartet-NJ', to infer an unrooted species tree on a
	given set of taxa $T$ from empirical distributions
	of unrooted quartet gene trees on all four-taxon
	subsets of $T$. In particular, quartet-NJ can be
	used to estimate a species tree on $T$ from
	distributions of gene trees on $T$. The quartet-NJ
	algorithm is conceptually very similar to classical
	neighbor joining, and its statistical consistency
	under the multispecies coalescent model is proven by
	a variant of the classical `cherry picking'-theorem.
	In order to demonstrate the suitability of
	quartet-NJ, coalescent processes on two different
	species trees (on five resp.\ nine taxa) were
	simulated, and quartet-NJ was applied to the
	simulated gene tree distributions. Further,
	quartet-NJ was applied to quartet distributions
	obtained from multiple sequence alignments of 28
	proteins of nine prokaryotes. \end{abstract}

\tableofcontents

\section{Introduction}


As the amount of available sequence data from genomes or
proteoms rapidly grows, one is confronted with the problem
that, for a given set of taxa, tree topologies relating
homologous sequences of these taxa can (and in practice
often will) differ from the each other, depending on which
locus in the genome or which protein is used for
constructing the tree. Possible reasons for discordance
among gene trees and the species phylogeny include horizontal
gene transfer, gene duplication/loss, or incomplete lineage
sorting (deep coalescences). Thus one is confronted with the
task to determine the phylogeny of the taxa (the `species
tree') from a set of possibly discordant gene trees
(Maddison \cite{Maddison_97}, and Maddison and Knowles
\cite{Maddison_Knowles}).
\newline


Several methods have been proposed and are used to infer a
species tree from a set of possibly discordant gene trees:

(1) Declaring the most frequently observed gene tree to be
the true species tree was shown to be statistically
inconsistent under the multispecies coalescent model by
Degnan and Rosenberg \cite{Degnan_Discordance}, in cases
where branch lengths of the underlying species tree are
sufficiently small.

(2) A popular approach for species tree reconstruction is by
concatenating multiple alignments from several loci to one
large multiple alignment and construct a single `gene tree'
from this multiple alignment by standard methods. This
approach was pursued e.g.\ by Teichmann and Mitchison
\cite{Teichmann_Mitchison} and Cicarelli et al.\
\cite{Cicarelli}. However, also this method was shown to be
inconsistent under the multispecies coalescent by Degnan and
Kubatko \cite{Degnan_Concat}, if branch lengths on the
species tree are sufficiently short.

(3) Similarly, the concept of minimizing the number of `deep
coalescences' (Maddison \cite{Maddison_97}) was recently
shown to be statistically inconsistent under the
multispecies coalescent model by Than and Rosenberg
\cite{Than_Rosenberg}.

(4) On the other hand, it is well known from coalescent
theory that the probability distribution of gene trees, or
even the distributions rooted gene triplet trees,
on a species tree uniquely determine the species tree if
incomplete lineage sorting is assumed to be the only source
of gene tree discordance. Similarly (see Allman et al.
\cite{Allman_Unrooted}), \emph{the distributions of unrooted
quartet gene trees identify the unrooted species tree
topology.} However, as soon as triplet/quartet gene trees are
inferred from experimental data, their distributions will
differ from the theoretical ones, and this may lead to
conflicts among the inferred triplets/quartets on the
hypothetical species tree, a problem which is not straight
forward to resolve (`quartet-puzzeling' could be an approach
to resolve this, see Strimmer and von Haeseler
\cite{Haeseler_Quartet}). Also direct maximum likelihood calculations
using gene tree distributions are very problematic due to
the enormous number of possible gene trees on a given set of
taxa.

(5) Recently Liu and Yu \cite{Liu_Neighbor} have proposed to
use the `average internode distance' on gene trees as a
measure of distance between taxa on species trees and apply
classical neighbor joining to these distance data. They
prove that this reconstructs the underlying species tree in
a statistically consistent way.
\newline

In the present paper we propose a method to overcome many of
the difficulties discussed above, and in particular to make
the above mentioned result by Allman et al.\
\cite{Allman_Unrooted} accessible in practice. That is, we
describe a polynomial time algorithm (in the
number of taxa), which uses empirical distributions of
unrooted quartet gene trees as an input, to estimate the
unrooted topology of the underlying species tree in a
statistically consistent way. Due to the conceptual
similarity to classical neighbor joining we call this
algorithm `quartet neighbor joining', or briefly
`quartet-NJ'.

The fact that quartet-NJ uses quartet distributions as an
input makes it flexible in practice: Quartet gene tree
distributions can be obtained directly from sequence data
(as is described in the application to a prokaryote data set
in Section \ref{sect:AppSim} of this paper), or e.g.\ from
gene tree distributions, which is the type of input required
by Liu and Yu's `neighbor joining for species trees'
\cite{Liu_Neighbor}. Also, quartet-NJ is naturally capable
of dealing with multiple lineages per locus and taxon.

The paper is organized as follows: After a brief review of
the multispecies coalescent model and distributions of
quartet gene trees in Section \ref{sect:MSCM}, we
investigate in Section \ref{sect:CherryPicking} how to
identify cherries on a species tree. To this end we show how
to assign `weights' to unrooted quartet trees on the set of
taxa, using the distributions of quartet gene trees, and how
to define a `depth' for each pair of taxa. In analogy to
classical cherry picking we prove that under the
multispecies coalescent model any pair of taxa with maximal
depth is a cherry on the species tree. Moreover, we give an
interpretation of this theorem by the concept of `minimal
probability for incomplete lineage sorting', in analogy to
the concept of minimal evolution underlying classical
neighbor joining.

In Section \ref{sect:Algorithms} we translate our cherry
picking theorem into a neighbor joining-like procedure
(`quartet-NJ') and
prove that it reproduces the true unrooted species tree
topology as the input quartet distributions tend to the
theoretical quartet distributions. In other words,
quartet-NJ is statistically consistent.

Finally, in Section \ref{sect:AppSim} we apply the quartet-NJ
algorithm to data from coalescence simulations, as well as
to a set of multiple alignments from nine prokaryotes, in
order to demonstrate the suitability of quartet-NJ. In both
situations we consider only one lineage per locus and taxon.

\section{The multispecies coalescent model}\label{sect:MSCM}

We are going to present briefly the multispecies coalescent,
which models gene (or protein) tree evolution within a fixed
species tree. The material in this section is neither new
nor very deep. Rather this section is to be considered as a
reminder for the reader on the relevant definitions, as well
a suitable place to fix language and notation for the rest
of the paper. For a more detailed introduction to the
multispecies coalescent model we refer e.g.\ to Allman et al.\
\cite{Allman_Unrooted}.

\subsection{Overview}

Let us consider a set $T$ of $n$ taxa (e.g.\ species) and a
rooted phylogenetic (that is: metric with strictly positive
edge lengths and each internal node is trivalent) tree
$S$ on the taxon set $T$. The tree $S$
is assumed to depict the `true' evolutionary relationship
among the taxa in $T$, and the lengths of its internal
edges are measured in coalescence units (see below). This
tree is commonly called the \emph{species tree} on the taxon
set $T$.

It is a fundamental problem in phylogenetics to determine
the topology of this tree. Molecular methodes, however,
usually produce evolutionary trees for single loci within
the genome of the relevant species, and these \emph{gene
trees} will usually differ topologically from the species
tree (Degnan and Rosenberg \cite{Degnan_Discordance}). This
has several biological reasons, like horizontal gene
transfer, gene duplication/loss or incomplete lineage
sorting. The latter describes the phenomenon that to gene
lineages diverge long before the population actually
splits into two separate species, and in particular gene
lineages may separate in a different order than the species
do.

The multispecies coalescent model describes the probability
for each rooted tree topology to occur as the topology of a
gene tree, under the assumption that incomplete lineage
sorting is the only reason for gene tree discordance.

Notation: We adopt the common practice to denote taxa
(species) by lower case letters, while we denote genes (or
more general: loci) in their genome by capital letters.
E.g.\ if we consider three taxa $a,b,c\in T$, the letters
$A,B,C$ will denote a particular locus sampled from the
respeceive taxa. (By abuse of notation, we will identify the
leaf sets of both species \emph{and} gene trees with $T$).

Considering only a single internal branch of the species
tree, and two gene lineages within this branch going
backwards in time, we find that the probability that the two
lineages coalesce to a single lineage within time $\tau$ is
given by $$ P = 1 - \exp( - \tau ) $$ where $\tau$ is in
\emph{coalescence units} (that is number of generations
represented by the internal branch divided by the total
number of allele of the locus of interest, present in the
population). In particular, if the branch on the species
tree has length $d$, the probability that two lineages
coalesce within this branch is $1-\exp(-d)$.

From this, Nei
\cite{Nei_87} derives the following probabilities under the
multispecies coalescent model for a three taxon tree. Denote
the taxa by $a,b,c$, the corresponding loci by $A,B,C$, and
assume that the species tree has the topology $((a,b),c)$
with internal edge length $d>0$. Then the probability that a
gene tree sampled from these taxa has the topology
$((A,B),C)$ is equal to $$ 1 - \frac{2}{3}\exp(-d),$$ while
the other two topologies are observed with probability
$\frac{1}{3}\exp(-d)$.

In \cite{Degnan_Distributions} Degnan and Salter show how to
calculate the probabilities of a gene tree in a fixed
species tree on an arbitrary taxon set. It turns out that
these probabilities are polynomials in the unknonws
$\exp(-d_{E})$, where $d_{E}$ denotes the length of the edge
$E$ on the species tree. In particular, the multispecies
coalescent model is an algebraic statistical model.

\subsection{Quartet distributions}

In the following we consider four taxon species trees on the
taxon set $T=\lbrace a,b,c,d \rbrace$, and we use Newick
notation to specify (rooted) species trees. For instance
$S_{1} = (((a,b):x,c):y,d)$ denotes the
caterpillar tree with edge lengths $x$ and $y$ in coalescent
units. Moreover, unrooted quartets are denoted in the format
$(AB,CD)$, meaning the quartet gene tree which has cherries
$(AB)$ and $(CD)$.

Up to permutation of leaf labels there is only one
additional tree topology on four taxa, namely the balanced tree
$S_{2} = ((a,b):x,(c,d):y)$. For both of these two
species tree topologies it is not hard to calculate the
probabilities of the tree possible unrooted quartet gene
trees on $T$, and one obtains

\begin{lemma}[Allman et al.\ \cite{Allman_Unrooted}]
	For both species tree topologies $S =
	S_{1}$ and $S =
	S_{2}$ the probabilities of unrooted
	quartet gene trees on $T$ have the same form and are given by the formulas
	\begin{align}
		P( AB,CD ) &= 1-\frac{2}{3}\exp(-x-y) >
		\frac{1}{3},\\
		P( AC,BD ) &= \frac{1}{3}\exp(-x-y) < \frac{1}{3},\\
		P( AD,BC ) &= \frac{1}{3}\exp(-x-y) <
		\frac{1}{3}.
		\label{eqnQuartetProb}
	\end{align}
	\label{lemQuartetProb}
\end{lemma}

In particular, the gene quartet distribution determines the
\emph{unrooted} species tree topology, but not the position
of the root (Allman et al.).
The sum of the internal edge lengths of the rooted
species tree (resp. the length of the internal edge of the
unrooted species tree) is given by the formula

\begin{equation}
	d(ab,cd) := d = x + y = -\log( \frac{3}{2}(1-P(AB,CD))).
	\label{eqn:QuartetEdgeLength}
\end{equation}

Returning to the study of the multispecies coalescent on
species trees on arbitrary taxon sets, we deduce from the
above that the quartets displayed on the species tree
$S$ on the taxon set $T$ are exactly those which
appear with a probability bigger than $\frac{1}{3}$ for any
sampled locus. Hence, the distributions of quartets
determine uniquely the quartet subtrees which are displayed
by the true species tree $S$, and hence determine
the unrooted topology of $S$ (see Allman et al.\
\cite{Allman_Unrooted}).

In the following two sections we describe a neighbor joining
algorithm which makes this theoretical insight applicable in
practice, meaning that it yields a method to estimate
unrooted species trees from (empirical) distributions of
quartet gene trees, which is statistically consistent under
the multispecies coalescent model. Statistical consistency
of this algorithm will follow from the `cherry picking
theorem' below.

\section{A cherry picking theorem}\label{sect:CherryPicking}

Here we give a precise criterion, using theoretical
distributions of quartet gene trees under the multispecies
coalescent, to determine which pairs of taxa on the species
tree are cherries. This criterion can as well be applied to
estimated quartet distributions and thus enables us to
recursively construct a species tree estimate from observed
gene quartet frequencies in Section \ref{sect:Algorithms}.

\subsection{Depth of a pair and cherry picking}

As always, we consider a fixed species tree $S$ on
the taxon set $T$, and we denote by $P(IJ,KL) :=
P_{S}(IJ,KL)$ the probability that a gene tree on
$T$ displays the gene quartet tree $(IJ,KL)$.

Recall that by lower case letters $i,j,l,k$ we denote the
species from which the genes $I,J,K,L$ are sampled.
Regardless of whether $S$ contains the (species-)
quartet $(ij,kl)$, we can attach to it the following
numbers:

\begin{definition}[Weight of a quartet]\label{defn:Weight}
	The \emph{weight} of
	the quartet $(ij,kl)$ is defined as
	\begin{equation}
		w(ij,kl) = -\log( \frac{3}{2} ( P(IK,JL) + P(IL,JK) )).
		\label{eqn:Weight}
	\end{equation} If the taxa $i,j,k,l$ are not
	pairwise distinct, then we set $w(ij,kl) = 0$.
\end{definition}

Of course, if $i,j,k,l$ are pairwise distinct we may
equivalently write \begin{equation*} w(ij,kl) = -\log(
	\frac{3}{2} ( 1-P(IJ,KL) ) ).  \end{equation*}

\begin{lemma} If the quartet $(ij,kl)$ is displayed on the
	species tree $S$, then the weight
	$w(ij,kl)$ is precisely the length of the interior
	branch of the quartet $(ij,kl)$. Moreover, the other
	two weights, $w(ik,jl)$ and $w(il,jk)$, are less
	than zero.
\label{lem:WeightDistance} \end{lemma}

\begin{proof}
	The proof is immediate using equation
	\eqref{eqn:QuartetEdgeLength}: 
	If the quartet tree $(ij,kl)$ is displayed on the
	species tree $S$, then
\begin{multline}
w(ij,kl) = -\log(\frac{3}{2}( P(IK,JL) + P(IL,JK) )) =\\ = -\log( \frac{3}{2} \frac{2}{3}\exp(-d(ij,kl)) ) = d(ij,kl).
\end{multline}
Otherwise, if $(ij,kl)$ is not displayed on $S$, then one of the other two quartet trees on $\lbrace i,j,k,l \rbrace$ is displayed, say $(ik,jl)$, and we have
\begin{multline}
w(ij,kl) = -\log(\frac{3}{2} ( P(IK,JL) +
P(IL,JK) )) = \\
= -\log( \frac{3}{2} ( 1- \frac{1}{3} \exp(-d(ik,jl))) ) < -\log(
\frac{3}{2}\frac{2}{3} ) = 0,
\end{multline}
since $d(ik,jl) > 0$.
\end{proof}

This little lemma motivates the following definition and a
cherry picking theorem which is formulated and proved in
close analogy to the `classical' one (see Saito and Nei
\cite{Nei_Neighbor} for the original publication, or Pachter
and Sturmfels \cite{Pachter_Sturmfels} for a presentation
analogous to ours).

\begin{definition}[Depth of a pair]\label{defn:Depth}
	For any pair of taxa $i,j\in T$ we define the
	\emph{depth} of $(i,j)$ to be the number
	\begin{equation}
		f(i,j) = \displaystyle\sum_{k,l\in T}
		w(ij,kl).
		\label{eqnDepthPair}
	\end{equation}
	(Recall that if $\lbrace i,j\rbrace\cap\lbrace
	k,l\rbrace \neq \emptyset$ then $w(ij,kl)=0$.)
\end{definition}

\begin{theorem}[Cherry Picking]\label{thm:CherryPicking}
	If a pair of taxa $i,j\in T$
	has maximal depth $f(i,j)$, then $(ij)$ is a cherry on the
	species tree $S$.
\end{theorem}

As mentioned above, the proof of this theorem is parallel to
the proof of the classical cherry picking result as
presented in \cite{Pachter_Sturmfels}. We state
it here for sake of completeness of our exposition.

\begin{proof}
	As a first step, we introduce the auxiliary values
\begin{equation}
	v(ij,kl) := \left\lbrace \begin{array}{cl}
		w(ij,kl) & \text{if } (ij,kl) \text{ is
		displayed on }S, \\
		0 & \text{else,}
	\end{array}\right.
	\label{eqn:defv}
\end{equation}
for every four taxon subset $\lbrace i,j,k,l \rbrace \subset
T$, as well as
\begin{equation}
	g(i,j) := \displaystyle\sum_{k,l}v(ij,kl)
	\label{eqn:defg}
\end{equation}
for every pair of taxa. Obviously, if $(ij)$ is a cherry on
$S$, then $w(ij,kl) = v(ij,kl)$ for all $k,l\in
T$, and hence also $g(i,j) = f(i,j)$. In general, for any
pair of taxa $(i,j)$ one has $g(i,j) \geq f(i,j)$ by Lemma
\ref{lem:WeightDistance}.

We will now assume that the pair $(i,j)$ does \emph{not}
form a cherry on the species tree $S$ and prove
that in this case there exists a cherry $(pq)$ on
$S$ such that the following inequalities hold:
\begin{equation}
	f(p,q) = g(p,q) > g(i,j) \geq f(i,j).
	\label{eqn:Inequal}
\end{equation}
From this we deduce the claim of the theorem: If $(ij)$ is
not a cherry on $S$, then it does not have maximal
depth.

\begin{proof}[Proof of the claim] We have to find a cherry
	$(pq)$ on $S$ such that indeed $g(p,q) >
	g(i,j)$ holds.
As we assume $i$ and $j$ not to form a cherry on $S$, the
unique path connecting $i$ and $j$ crosses at least two
interior nodes (symbolized as black dots in figure \ref{Tree1}. We
denote the number of internal nodes on this path by $r\geq
2$. To each $s=1,\dotsc,r$ a rooted binary tree $T_{j}$ is
attached (see Figure \ref{Tree1}).
\begin{figure}
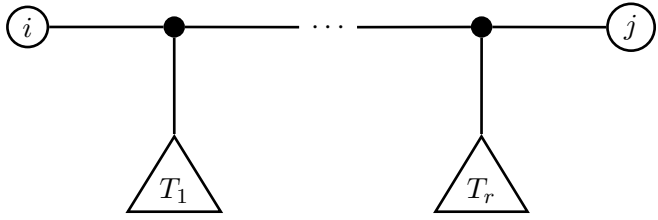

	\caption{The nodes $i$ and $j$ do not form a cherry.}
	\begin{center}
$
\psmatrix[mnode=circle, linewidth=1pt, rowsep=1cm, colsep=1cm]
i & [fillstyle=solid,fillcolor=black] & [linestyle=none]\cdots & [fillstyle=solid,fillcolor=black] & j \\
& [mnode=tri]T_{1} & & [mnode=tri]T_{r} 
\ncline{1,1}{1,2}
\ncline{1,2}{1,3}
\ncline{1,3}{1,4}
\ncline{1,4}{1,5}
\ncline{1,2}{2,2}
\ncline{1,4}{2,4}
\endpsmatrix
$
\end{center}
	\label{Tree1}
\end{figure}
By interchanging $i$ and $j$, if necessary, we may assume
that the number of taxa in $T_{1}$ is less or equal than the
number of taxa in $T_{r}$.

First we consider the case that the subtree $T_{1}$ consists
of a single leaf $i'$. In this case the pair $(i,i')$ is a
cherry on $S$, and it is easy to see that indeed
$g(i,i')>g(i,j)$.

Second, if $T_{1}$ has more than one leaf, let choose a
cherry $(pq)$ in $T_{1}$ (every rooted binary tree with more
than one leaf has at least one cherry!).
Now we decompose the difference $g(p,q)-g(i,j)$ into six
sums which will be treated separately:
\begin{equation}
	g(p,q) - g(i,j) = S_{1} + \dotsc + S_{5},
\end{equation} where
\begin{align*}
	S_{1} = & \displaystyle\sum_{k,l\in T_{m}, m=2,\dotsc,r}
	(v(pq,kl) - v(ij,kl)), & \\
	S_{2} = & \displaystyle\sum_{k\in T_{m} \setminus \lbrace
	p,q \rbrace, l\in T_{m'} \setminus \lbrace p,q
	\rbrace,
	m,m'=1,\dotsc,r, m\neq m'} (v(pq,kl) - v(ij,kl)), \\
	S_{3} = & \displaystyle\sum_{k,l\in T_{1} \setminus \lbrace
	p,q\rbrace} (v(pq,kl) - v(ij,kl)), \\
	S_{4} = & \displaystyle\sum_{k\in T_{m}, m=2,\dotsc,r}
	(v(pq,ki) + v(pq,kj) - v(ij,kq) - v(ij,kp)), \\
	S_{5} = & \displaystyle\sum_{k\in T_{1}\setminus\lbrace
	p,q\rbrace}
	(v(pq,ki) + v(pq,kj) - v(ij,kq) - v(ij,kp))\quad + \\
	& + (-v(pq,ij) + v(ij,pq)).
	\label{eqn:cherry_noncherry}
\end{align*}
Inspection of the tree in Figure \ref{Tree1} immediately
shows that each summand in $S_{1}$ is
positive, and more precisely, greater or equal to
$d(ij,pq)$. Hence $S_{1} > \binom{\lvert T_{r}\rvert}{2}d(ij,pq)$.
In $S_{2}$, the expressions $v(ij,kl)$ all vanish,
hence this sum is positive.
\begin{figure}
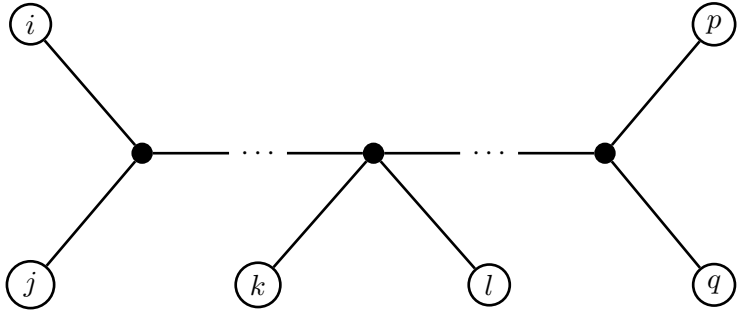

	\caption{Part of the tree: leaves $i,j$ and subtree
	$T_{1}$ with cherry $(pq)$ and leaves $k,l$.}
	\begin{center}
$
\psmatrix[mnode=circle, linewidth=1pt, rowsep=1cm, colsep=1cm]
i & & & & & & p \\
& [fillstyle=solid,fillcolor=black] & [linestyle=none]\cdots &
[fillstyle=solid,fillcolor=black] &
[linestyle=none]\cdots &
[fillstyle=solid,fillcolor=black] \\
j & & k & & l & & q 
\ncline{1,1}{2,2}
\ncline{3,1}{2,2}
\ncline{2,2}{2,3}
\ncline{2,3}{2,4}
\ncline{2,4}{2,5}
\ncline{2,5}{2,6}
\ncline{2,6}{1,7}
\ncline{2,6}{3,7}
\ncline{3,3}{2,4}
\ncline{3,5}{2,4}
\endpsmatrix
$
\end{center}
	\label{Tree2}
\end{figure}
The third sum $S_{3}$ might be negative - but not too negative: The
situation is illustrated by Figure \ref{Tree2} (which arises
from the tree in Figure \ref{Tree1} by deleting the
irrelevant subtrees $T_{2},\dotsc,T_{r}$), whose
inspection shows that each for each summand $v(pq,kl) -
v(ij,kl)$ we have
\begin{equation*}
	v(pq,kl) - v(ij,kl) = d(pq,kl) - d(ij,kl) \geq
	-d(ij,pq).
\end{equation*}
Thus we see $S_{3} \geq -\binom{\lvert
T_{1}\rvert - 2}{2}d(ij,pq) > -\binom{\lvert T_{r}\rvert}{2}d(ij,pq)$,
whence $S_{1} + S_{2} + S_{3}$ is positive.

Now we treat the frouth sum: If $k$ is a node in one of
the subtrees $T_{2},\dotsc,T_{r}$, then $v(ij,kq)=v(ij,kp) =
0$, so $S_{4}$ is non-negative. More precisely we have
$S_{4} \geq 2\binom{\lvert T_{r} \rvert}{2}v(pq,ij)$.
On the other hand, the sum $S_{5}$ is greater or equal
$-2\binom{\lvert T_{1}\rvert - 2}{2}v(ij,pq)\geq
-2\binom{\lvert T_{r} \rvert}{2}v(pq,ij)$, whence also
$S_{4} + S_{5}$ is positive.
In total we have found that in any case $g(p,q)>g(i,j)$,
which proves our claim.
\end{proof}

As explained above, this suffices to prove the cherry
picking theorem.
\end{proof}

\subsection{Statistical consistency of cherry picking}

In practical applications one will not know the precise
probability of each quartet gene tree $(AB,CD)$. Hence one
has to use e.g.\ relative frequencies as estimates. Let us
denote by $r(AB,CD)$ the (experimentally determined)
relative frequency of the gene quartet $(AB,CD)$. In analogy
to Definitions \ref{defn:Weight} and \ref{defn:Depth} we
make the following

\begin{definition}
	The \emph{empirical weight} of the
	quartet $(ij,kl)$ of taxa is defined as
	\begin{equation} \tilde{w}(ij,kl) = -\log( \frac{3}{2}
		(r(IK,JL) + r(IL,JK)) ),
		\label{eqn:empWeight}
	\end{equation}
	Moreover, the \emph{empirical depth} of a pair
	$(i,j)$ of taxa is defined as
\begin{equation}
	\tilde{f}(i,j) = \displaystyle\sum_{k,l\in
	T}\tilde{w}(ij,kl).
	\label{eqn:empDepth}
\end{equation}
\end{definition}

Under the multispecies coalescent model, the relative
frequency $r(IJ,KL)$ of a quartet gene tree $(IJ,KL)$ is a
statistically consistent estimator for its probability
$P(IJ,KL)$. Hence the \emph{empirical} depth
$\tilde{f}(i,j)$ of a pair of taxa $i,j$ is a statistically
consistent estimator for $f(i,j)$. In particular this means:
\emph{inferring a pair $(i,j)$ of taxa, where $\tilde{f}$ is
maximal, as a cherry on the species tree is statistically
consistent.} More precisely, we have

\begin{corollary}\label{cor:consistent}
	Let $N$ be the number of loci sampled from each
	taxon in $T$ and denote by $C_{N}$ the set of pairs
	of taxa at which $\tilde{f}$ attains its maximum.
	If the number of genes $N$ tends to infinity, then
	the probability that $C_{N}$ contains a pair which
	is not a cherry on the true species tree approaches
	zero.
\end{corollary}

\begin{proof} As there are only finitely many taxa in $T$,
	there is a strictly positive difference between the
	maximum value of $f$ on $T^{2}$ and its second
	biggest value. Since moreover $f$ and $\tilde{f}$
	are continuous, there exists an $\delta > 0$ such
	that whenever $\lvert r(IJ,KL) - P(IJ,KL) \rvert <
	\delta$, $\tilde{f}(p,q)$ is maximal if and only if
	$f(p,q)$ is maximal, and the claim follows from
	Theorem \ref{thm:CherryPicking}. But the probability
	that $\lvert r(IJ,KL) - P(IJ,KL) \rvert < \delta$
	approaches 1 as $N$ grows, for any choice of
	$\delta$.
\end{proof}

\subsection{Minimal probability for incomplete
lineage sorting}

Recall that classical neighbor joining is a greedy algorithm
which in each `cherry picking' step declares a pair of taxa
to be neighbors if this minimizes the sum of branch lengths
in the refined tree resulting from this step (Saito and Nei
\cite{Nei_Neighbor}). In other words, classical cherry
picking and neighbor joining is guided by the principle of
\emph{minimal evolution}. It turns out that our cherry
picking result in Theorem \ref{thm:CherryPicking} can be
interpreted in a similar fashion. 

Let us make the following ad-hoc definition.
\begin{definition}
	Let $X_{1},\dotsc,X_{n}$ be independent random
	variables with values in $\lbrace 0,1\rbrace$, and
	let $p_{i}$ be the probability of $\lbrace X_{i} =
	1\rbrace$ for each $i=1,\dotsc,n$.
	We call the geometric mean
	\begin{equation}
		\label{eqn:averageProb}
		\bar{p} = \sqrt[n]{p_{1}\dotsb p_{n}}
	\end{equation}
	the \emph{average probability} of the random
	experiments $X_{1},\dotsc,X_{n}$ giving the result
	$1$.
\end{definition}

Let us now consider what the cherry picking Theorem
\ref{thm:CherryPicking} does with a set of taxa $T$,
initially arranged as a star-like tree (see Figure
\ref{fig:CherryPickingTree}). Consider two taxa
$i\neq j\in T$ fixed and let $q(ij,kl) = P(IK,JL) + P(IL,JK)$.
On a species tree on $T$ which displays the cherry $(ij)$, this
is the probability that a gene quartet tree sampled from the
taxa $i,j,k,l$ differs from the species quartet tree
$(ij,kl)$. Let
$$
\bar{q}(i,j) =
\sqrt[\binom{n-2}{2}]{\prod_{k,l}q(ij,kl)}
$$
be the average probability of discordance of a
sampled gene quartet tree with the corresponding quartet on
the species tree, for a set of four taxa containing $i$ and
$j$. We will also call this number the `average probability
of incomplete lineage sorting' for quartets containing the
taxa $i$ and $j$. With this terminology, the cherry picking
Theorem \ref{thm:CherryPicking} can be phrased as follows.

\begin{theorem}
	A pair of taxa $i,j\in T$ is a cherry on the true
	species tree $S$ if it minimizes the
	average probability for incomplete lineage sorting
	for quartets containing the taxa $i$ and $j$.
	\label{thm:avProbIncomplete}
\end{theorem}

\begin{proof}
	Using Theorem \ref{thm:CherryPicking} we only have
	to show that $\bar{q}(i,j)$ is minimal if and only
	if $f(i,j)$ is maximal. But this follows from the
	fact that
	$$f(i,j) = -\log \bar{q}(i,j) + C,$$
	where $C$ is a constant independent of $i,j$.
\end{proof}


\begin{figure}
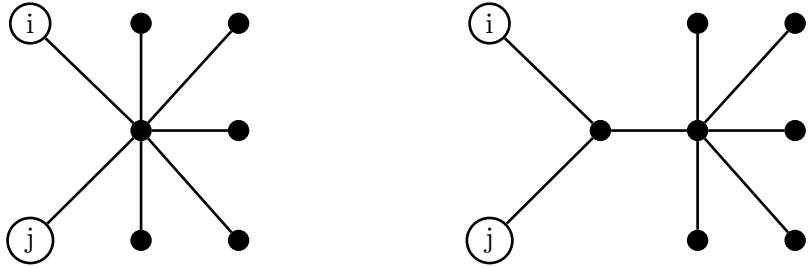

	\caption{One cherry picking step.}
	\begin{center}
$
\begin{tabular}{lm{2cm}r}
\psmatrix[mnode=circle, linewidth=1pt, rowsep=1cm, colsep=1cm]
i &
[fillstyle=solid,fillcolor=black] &
[fillstyle=solid,fillcolor=black] \\ 
 &
[fillstyle=solid,fillcolor=black] &
[fillstyle=solid,fillcolor=black] \\ 
j &
[fillstyle=solid,fillcolor=black] &
[fillstyle=solid,fillcolor=black] 
\ncline{2,2}{1,1}
\ncline{2,2}{1,2}
\ncline{2,2}{1,3}
\ncline{2,2}{2,3}
\ncline{2,2}{3,1}
\ncline{2,2}{3,2}
\ncline{2,2}{3,3}
\endpsmatrix
& &
\psmatrix[mnode=circle, linewidth=1pt, rowsep=1cm, colsep=1cm]
i & &
[fillstyle=solid,fillcolor=black] &
[fillstyle=solid,fillcolor=black] \\ 
& [fillstyle=solid,fillcolor=black] &
[fillstyle=solid,fillcolor=black] &
[fillstyle=solid,fillcolor=black] \\ 
j & &
[fillstyle=solid,fillcolor=black] &
[fillstyle=solid,fillcolor=black] 
\ncline{2,2}{2,3}
\ncline{2,2}{1,1}
\ncline{2,2}{3,1}
\ncline{2,3}{1,3}
\ncline{2,3}{1,4}
\ncline{2,3}{2,4}
\ncline{2,3}{3,4}
\ncline{2,3}{3,3}
\endpsmatrix
\end{tabular}
$
\end{center}
	\label{fig:CherryPickingTree}
\end{figure}

\section{Quartet neighbor joining
algorithms}\label{sect:Algorithms}

In this section we discuss how the above cherry picking
result can be used to design an algorithm which estimates a
species tree from (observed) gene quartet tree frequencies.

\subsection{A naive neighbor joining algorithm}

The first version of our neighbor joining algorithm
reconstructs the underlying species tree from the
(theoretical) gene quartet tree distributions.

\begin{algorithm}\label{alg:Naive} \textbf{Input}: A set of
	taxa $T=\lbrace t_{1},\dotsc,t_{n}\rbrace$
	containing at least four elements, and for each
	quartet $Q=(ab,cd)$ with $a,b,c,d\in T$ the
	probability $r(Q) := P_{S}(Q)$ of the
	quartet $Q$ under the multispecies coalescent model
	on the species tree $S$.

	\textbf{Output}: The species tree $S$.

	\textbf{Step 0.} Set $S$ to be the graph
	with vertex set $T$ and with empty edge set.

	\textbf{Step 1.} If $\lvert T\rvert \leq 3$ go to
	step 4.

	\textbf{Step 2.} For each unordered pair $\lbrace
	i,j \rbrace \subset T$ calculate the depth $f(i,j)$.
	Let $i,j \in T$ be a pair of taxa with maximal
	depth. Add a node $x$ to $S$ and draw
	edges from $i$ to $x$ and from $j$ to $x$. Replace
	$T$ by $T\cup\lbrace x\rbrace \setminus\lbrace
	i,j\rbrace$.
	
	\textbf{Step 3.} For each quartet $(xa,bc)$ with
	$a,b,c\in T\setminus\lbrace x\rbrace$ set
	\begin{equation*} r'{(xa,bc)} = \frac{1}{2}(r(ia,bc)
		+ r(ja,bc)), \end{equation*} and for all
		quartets $(ab,cd)$ which do not contain $x$
		set $r'(ab,cd) = r(ab,cd)$. Replace $r$ by
		$r'$. Go to step 1.

	\textbf{Step 4.} If $\lvert T\rvert = 3$, then add a
	vertex to $S$ and draw edges from this
	new vertex to the three elements of $T\subset
	vertices(S)$. \qed \end{algorithm}

\begin{remark} Each calculation of $f(i,j)$ requires
	$O(\lvert T\rvert^{2})$ logarithms and additions.
	Repeating this for every pair $i,j$ of taxa gives a
	computational complexity of $O(\lvert T\rvert^{4})$
	for step 2. This is also the complexity of one
	iteration of the algorithm. Since each step will
	identify exactly one cherry, we need $O(\lvert
	T\rvert^{5})$ iterations to reconstruct $S$, which
	gives a total complexity of $O(\lvert T\rvert^{5})$
	for Algorithm \ref{alg:Naive}. It is rather
	straightforward to improve step 2 so that the
	algorithm requires only $O(\lvert T\rvert^{4})$
	arithmetic operations (see Subsection
	\ref{subsect:complexity}).  \end{remark}

Using the result of Theorem \ref{thm:avProbIncomplete}, we
may summarize: \emph{With Algorithm \ref{alg:Naive} we have
described a greedy algorithm which inferes in each step the
cherry which requires minimal incomplete lineage sorting.}
Further, the fact that cherry picking is statistically
consistent implies that also species tree estimation by
Algorithm \ref{alg:Naive} is statistically consistent:

\begin{corollary}[Statistical consistency of Algorithm
	\ref{alg:Naive}] Let $N$ be the number of loci sampled from
	each taxon in $T$ and denote by $S_{N}$ the estimate
	for the species tree produced by Algorithm
	\ref{alg:Naive}. If the number of genes $N$ goes to
	infinity, then the probability that $S_{N}$ equals
	the true species tree $S$ approaches 1.
\end{corollary}

\begin{proof}
	This follows by a repeated application of Corollary
	\ref{cor:consistent}.
\end{proof}

In practical applications two problems may occur using
Algorithm \ref{alg:Naive}: First, there might be several
non-disjoint pairs of taxa with maximal depth. This could be
resolved by chosing one of these pairs randomly. A more
serious problem occurs when many of the empirical gene
quartet distributions are 0. This will be discussed in the
following subsection.

\subsection{Perturbing weights and quartet-NJ}

Algorithm \ref{alg:Naive} works well if one uses as input
the theoretical probabilities for quartet gene trees. It
also works well, if one uses relative quartet frequencies
which are `very close' to the probabilities $P(IJ,KL)$ (and
in particular non-zero). In other cases, the result might
be problematic, as we are going to explain now.

Imagine that for some reason (e.g.\ very long branch lenghts
on the species tree) the observed quartet gene trees reflect
very precisely the topology of the species tree. This might
mean in the extreme case that for every four-taxon subset
$\lbrace I,J,K,L\rbrace$ one observes only this very quartet
gene tree which is displayed also by the (unknown) species
tree. In some sense this situation should be optimal, since
the observed gene quartet trees fit together without
conflict and thus yield an unambiguous estimate for the
species tree.

However, let us run Algorithm \ref{alg:Naive} with the
empirical depth function $\tilde{f}$ in place of $f$ and
consider any pair $i,j$ of taxa. Regardless of the choice of
$i$ and $j$, there exist $k,l\in T$ such that $(ij,kl)$ is a
quartet on the species tree, and hence by our assumption
$r(IJ,KL) = 1$. Calculating the empirical depth of $(i,j)$
yields \begin{equation}\label{eqn:infty} \tilde{f}(i,j) =
	\displaystyle\sum_{k,l\in T}
	-\log(\frac{3}{2}(1-r(IJ,KL)) = +\infty.
\end{equation} Thus \emph{every} pair $(i,j)$ maximizes
$\tilde{f}$, and so in each iteration of Algorithm
\ref{alg:Naive} the choice of a cherry is completely
arbitrary. So this procedure fails in a situation, which has
to be considered the `easiest' possible in some obvious
sense.

A possible solution to this problem is to perturb the
arguments in the logarithms in equation \eqref{eqn:infty} by
a small number $\epsilon > 0$. In other words, we fix
$0\leq \epsilon \ll 1$ and calculate, for each pair of taxa
$i,j$, the `perturbed depths'
\begin{align}\label{eqn:finite}
	f_{\epsilon}(i,j) &= \displaystyle\sum_{k,l\in T}
	-\log(\frac{3}{2}((1 + \epsilon) - P(IJ,KL)) <
	\infty,\\
	\tilde{f}_{\epsilon}(i,j) &= \displaystyle\sum_{k,l\in T}
	-\log(\frac{3}{2}((1 + \epsilon) - r(IJ,KL)) <
	\infty.
\end{align}
Obviously, for $\epsilon \to 0$ we have $f_{\epsilon}(i,j)
\to f_{0}(i,j)=f(i,j)$ and $\tilde{f}_{\epsilon}(i,j) \to
\tilde{f}_{0}(i,j) = \tilde{f}(i,j)$, for all pairs
$(i,j)\in T^{2}$. From this we obtain the following easy

\begin{lemma} Assume that the theoretical probabilities
	$P(IJ,KL)$ are known for each four taxon subset
	$\lbrace I,J,K,L\rbrace\subset T$ and are used as
	input for Algorithm \ref{alg:Naive}. Then there
	exists a constant $c>0$ such that Algorithm
	\ref{alg:Naive} computes the same results with
	$f_{\epsilon}$ in place of $f$, for each $0\leq
	\epsilon < c$. In other words, if $S_{\epsilon}$ is
	the species tree inferred by Algorithm
	\ref{alg:Naive} with $f_{\epsilon}$ in place of $f$,
	then the limit $\underset{\epsilon \to 0}\lim
	S_{\epsilon}$ exists and is equal to $S=S_{0}$.
	\label{lem:epsilon} \end{lemma}

\begin{proof} This follows from the fact that $S$ is a
	binary tree, that there are only finitely many
	values of $f$ and that $f_{\epsilon}(i,j)$ depends
	continuously on $\epsilon$.
\end{proof}

This suggests that for practical applications it will be
reasonable to fix $\epsilon>0$ `small enough', and run
Algorithm \ref{alg:Naive} with the perturbed empirical depth
function $\tilde{f}_{\epsilon}$ in place of $f$ in order to
avoid `infinite depths' as in the discussion at the
beginning of this subsection. Hence we obtain the following
modification of Algorithm \ref{alg:Naive}.

\begin{algorithm}[Quartet Neighbor
	Joining, Quartet-NJ]\label{alg:improved}
	\textbf{Input}: A finite set of taxa $T=\lbrace
	t_{1},\dotsc,t_{n}\rbrace$, a `small' constant
	$\epsilon > 0$, and for each four taxon subset
	$\lbrace a,b,c,d\rbrace \subset T$ the relative quartet
	gene tree frequencies $r{(AB,CD)}$, $r{(AC,BD)}$,
	$r{(AD,BC)}$.

	\textbf{Output}: An unrooted tree
	$\tilde{S}_{\epsilon}$ on $T$ which is our estimate
	for the topology of the species tree $S$ on $T$.

	\textbf{Algorithm}: Run Algorithm \ref{alg:Naive}
	with the depth function $f(i,j)$ substituted by the
	empirical depth $\tilde{f}_{\epsilon}(i,j)$. The result
	of this calculation is $\tilde{S}_{\epsilon}$.
\qed \end{algorithm}

Indeed, if we apply this modified algorithm to the
problematic situation at the beginning of this
subsection, we will obtain (for any choice of $\epsilon$) a
fully resolved correct estimate for the species tree $S$. Of
course, this algorithm has the same complexity as Algorithm
\ref{alg:Naive}.
\newline

We want to claim that, for sufficiently small values of
$\epsilon$, quartet-NJ produces a asymptotically correct
estimate of the true species tree on $T$. To this end, we
first prove

\begin{lemma} There exists a constant $c>0$ such that for
	all $0<\epsilon<c$ the trees $\tilde{S}_{\epsilon}$
	are equal. In other words, the limit
	$\underset{\epsilon \to 0}\lim \tilde{S}_{\epsilon}$
	exists.
\label{lem:epsilon_empirical} \end{lemma}

\begin{proof} Since there are only finitely many cherries to
	pick, we may restrict our considerations to the
	picking of the first cherry. We have to distinguish
	two cases: First assume that there are taxa $i,j$
	such that $\tilde{f}(i,j) = \infty$, i.e. the set
	$Inf = \lbrace (i,j) \suchThat \tilde{f}(i,j) =
	\infty \rbrace$ is nonempty. If $\epsilon$ is small
	enough, then the maximum of the values of
	$\tilde{f}_{\epsilon}$ is attained at one of the
	pairs in $Inf$. At which of those pairs the maximum
	is attained then only depends on the number of
	summands in $\tilde{f}_{\epsilon}(i,j)$ which
	approach infinity as $\epsilon$ tends to zero. This
	number is clearly independent of $\epsilon$, whence
	the set of potential cherries to pick is independent
	of $\epsilon$.

	In the second case, there are no elements in the set
	$Inf$. This means that each of the perturbed
	empirical depths $\tilde{f}_{\epsilon}(i,j)$
	approaches the \emph{finite} value $\tilde{f}(i,j)$
	as $\epsilon$ tends to zero. This again means that,
	for $\epsilon$ small enough, the pair which
	maximizes $\tilde{f}_{\epsilon}$ does not depend on
	$\epsilon$.
\end{proof}

Combining this with Lemma \ref{lem:epsilon} we obtain the
desired result.

\begin{theorem}[Statistical consistency of quartet-NJ for
	small $\epsilon$] The limit $\underset{\epsilon\to
	0}{\lim}(\tilde{S}_{\epsilon})$ exists and is a
	statistically consistent estimate for the true
	species tree $S=S_{0}$. In particular, quartet-NJ
	ist statistically consistent. \label{thm:main}
\end{theorem}

\begin{proof}
	The existence of the limit in the theorem is
	established by Lemma \ref{lem:epsilon_empirical}. It
	remains to prove that it is a statistically
	consistent estimate for $S$.
	Recall that the definition of the function
	$\tilde{f}_{\epsilon}(i,j)$ depends on the relative
	frequencies $r(IJ,KL)$. If we consider the
	probabilities $P(IJ,KL)$ known and fixed and
	moreover fix the pair of taxa $(i,j)$, then we may
	consider the function
	\begin{equation}
		F(i,j) = \tilde{f}_{\epsilon}(i,j) -
		f_{\epsilon}(i,j):\quad
		\mathbb{R}_{>0}^{\binom{n}{4}}\times
		\mathbb{R} \to \mathbb{R},
		\label{eqn:difference}
	\end{equation}
	depending on the `variables' $r(IJ,KL)$ and
	$\epsilon$. This is a continuous function which
	vanishes on the line $\lbrace (P(IJ,KL))_{(IJ,KL)})
	\rbrace \times \mathbb{R} \subset
	\mathbb{R}^{\binom{n}{4}}\times\mathbb{R}$.
	Thus for every $\delta$ there exists an open
	ball centered at $((P(IJ,KL))_{(IJ,KL)}, 0)$ which
	is mapped to $(-\delta, \delta)\subset\mathbb{R}$ by
	$F(i,j)$. In particular, there exists a positive constant
	$c > 0$ such that for every $(i,j)$ we have
	$\lvert \tilde{f}_{\epsilon}(i,j) -
	f_{\epsilon}(i,j)\rvert < \delta$ if $\lvert r(IJ,KL) -
	P(IJ,KL)\rvert < c$ and $\epsilon < c$.
	If we chose $\delta$ small enough, we thus may
	conclude that $\tilde{S}_{\epsilon} = S_{\epsilon}$
	for every $\epsilon < c$ and for all relative
	frequencies satisfying $\lvert r(IJ,KL) -
	P(IJ,KL)\rvert < c$. Moreover, by Lemma
	\ref{lem:epsilon} we may assume, by
	reducing $c$ if necessary, that $S_{\epsilon} = S$
	for every $\epsilon < c$.

	Taking this together we obtain that, provided $\lvert r(IJ,KL) -
	P(IJ,KL)\rvert < c$ for every quartet $(ij,kl)$,
	$\underset{\epsilon\to 0}\lim \tilde{S}_{\epsilon} = S$.
	Since $r(IJ,KL)$ is a statistically consistent
	estimate for $P(IJ,KL)$, this proves that
	$\underset{\epsilon\to 0}\lim \tilde{S}_{\epsilon}$ is a
	statistically consistent estimate for $S$.
\end{proof}

The question of how to find an appropriate $\epsilon >
0$ to run the neighbor joining algorithm will of
course depend on the special problem instance and is
left open here.

\subsection{Reduction of computational
complexity}\label{subsect:complexity}

Here we briefly mention a complexity reduction
for the quartet neighbor joining algorithm. Since Algorithm
\ref{alg:Naive} and \ref{alg:improved} are completely
equivalent in this respect, we formulate this reduction only
with the simpler notation of Algorithm \ref{alg:Naive}.

We consider step 2 in Algorithm \ref{alg:Naive}. In our
present formulation, this step calculates the depth $f(i,j)$
by adding up $O(\lvert T\rvert^{2})$-many quartet weights
for each pair of taxa $i,j$.
However, most of these quartet weights will not have changed
during the last iteration of the algorithm.

Assume that in the previous iteration the cherry
$(i_{0}j_{0})$ was identified and joined by a new vertex
$x$. Then for any pair of taxa $i,j$ which are still present in
$T$ we may calculate the new depth $f_{new}(i,j)$ from the old one
by the formula
$$ f_{new}(i,j) = f(i,j) - \displaystyle\sum_{k\in T}(
w(ij,i_{0}k) + w(ij,j_{0}k) + \displaystyle\sum_{k\in T}
w_{new}(ij,kx).$$
This reduces the complexity of step 2 from $O(\lvert
T\rvert^{4})$ to $O(\lvert T\rvert^{3})$, and hence the
overall complexity of Algorithms \ref{alg:Naive} and
\ref{alg:improved} are reduced by this modification to
$O(\lvert T\rvert^{4})$.

\section{Application and simulations}\label{sect:AppSim}

\subsection{Application to a prokaryote data
set}\label{subsect:prokaryotes}

In order to test quartet neighbor joining on real data,
Algorithm \ref{alg:improved} was applied to a set of protein
sequences from nine prokaryotes, among them the two archaea
\emph{Archaeoglobus fulgidus} (AF) and \emph{Methanococcus
jannaschii} (MJ), as well as the seven bacteria
\emph{Aquifex aeolicus} (AQ), \emph{Borrelia burgdorferi}
(BB), \emph{Bacillus subtilis} (BS), \emph{Escherichia coli}
(EC), \emph{Haemophilus influenzae} (HI), \emph{Mycoplasma
genitalium} (MG), and \emph{Synechocystis sp.} (SS).

The choice of organisms follows Teichmann and Mitchison
\cite{Teichmann_Mitchison}, while the multiple sequence
alignments where taken out of the dataset used by Cicarelli
et al.\ in \cite{Cicarelli}. For each of the 28 protein
families in the list

\begin{enumerate}
	\item Ribosomal proteins, small subunits: S2, S3,
		S5, S7, S8, S9, S11, S12, S13, S15, S17, L1, L3,
		L5, L6, L11, L13, L14, L15, L16, L22;
	\item tRNA-synthetase: Leucyl-, Phenylalanyl-,
		Seryl-, Valyl;
	\item Other: GTPase, DNA-directed RNA polymerase
		alpha subunit, Preprotein translocalse
		subunit SecY,
\end{enumerate}

a distance matrix was calculated using BELVU
\cite{Sonnhammer} with `scoredist'-distance correction
(Sonnhammer and Hollich \cite{Sonnhammer}). For each
distance matrix, a set of quartet trees was inferred in the
classical way by finding, for each four taxon set $\lbrace
i,j,k,l\rbrace$, the unrooted quartet $(IJ,KL)$ which
maximizes the value
$$ d(I,K) + d(I,L) + d(J,K) + d(J,L) -
2d(I,J) - 2d(K,L),
$$
where $d(-,-)$ denotes the respective
entry in the distance matrix.

Algorithm \ref{alg:improved} was then run with the
parameter $\epsilon = 10^{-6}, \epsilon = 10^{-9}$, and $\epsilon
= 10^{-12}$ on the quartet distribution obtained from
analyzing all the 28 protein families above, and in a
second try on the quartet distributions obtained only by
the ribosomal proteins. The resulting tree topologies are
depicted in Figures \ref{fig:all} and \ref{fig:ribosomal},
respectively (The root of these trees is of course not
predicted by Algorithm \ref{alg:improved}. Rather it was
placed a posteriori on the branch which separates archaea from
bacteria on the unrooted output of the algorithm.) The
different choices of $\epsilon$ did not affect the result in
these calculations.
\newpage
\begin{figure}
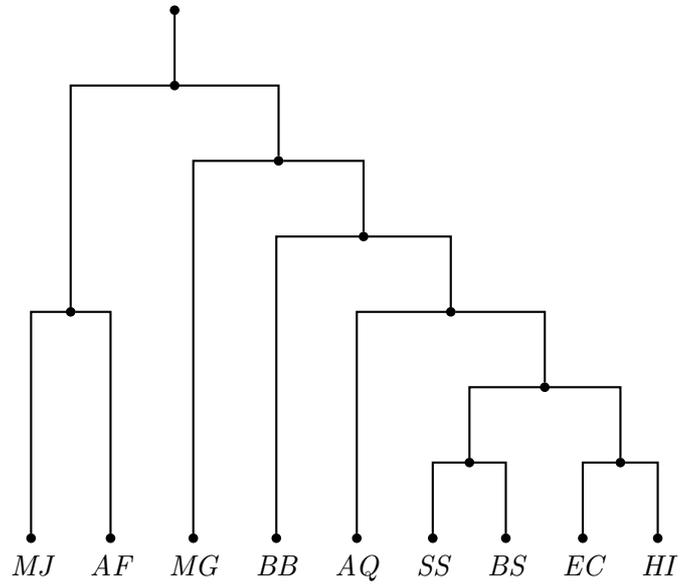

	\centering
\begin{newicktree}
	\nobranchlengths\setunitlength{1cm}
	\drawtree{((MJ:3,AF:3):3,(MG:5,(BB:4,(AQ:3,((SS:1,BS:1):1,(EC:1,HI:1):1):1):1):1):1):1;}
\end{newicktree}
	\caption{The neighbor joining species tree topology for the
	nine prokaryotes, using multiple alignments for all
	28 protein families.}
	\label{fig:all}
\end{figure}
\begin{figure}
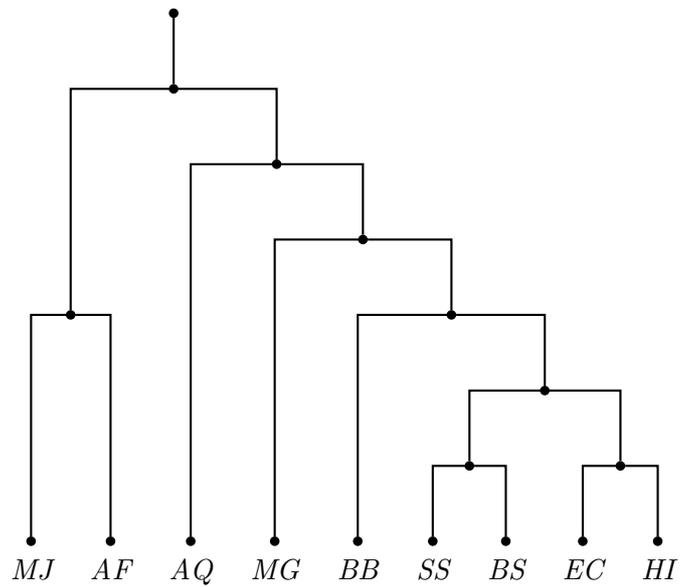

	\centering
\begin{newicktree}
	\nobranchlengths\setunitlength{1cm}
	\drawtree{((MJ:3,AF:3):3,(AQ:5,(MG:4,(BB:3,((SS:1,BS:1):1,(EC:1,HI:1):1):1):1):1):1):1;}
\end{newicktree}
	\caption{The neighbor joining species tree topology for the
	nine prokaryotes, using only the multiple
	alignments of the ribosomal proteins above.}
	\label{fig:ribosomal}
\end{figure}
\clearpage

\subsection{Simulations with Mesquite}

As a first test of performance of Algorithm
\ref{alg:improved} two series of simulations were performed
as follows. For a certain choice of a species tree $S$ on a
set of taxa $T$, a coalescent process was repeatedly
simulated using the coalescence package of Mesquite
\cite{mesquite}, \cite{mesquite_coal}. Each simulation
yielded a set of gene trees on $T$, and from these the
frequency of each unrooted quartet gene tree with leaves in
$T$ was determined. These frequencies where then submitted
to Algorithm \ref{alg:improved} and the resulting (unrooted)
tree was compared with the (unrooted version of the) species
tree $S$. We report here the proportion of correct
inferences of the unrooted species tree in the different
situations.
\newline

In the first series of simulations the underlying species
tree was the 5-taxon caterpillar tree $(a,(b,(c,(d,e))))$,
with all internal branch lengths set equal (of course, the
length of the pending edges does not have an impact, as we
consider only one lineage per taxon). Algorithm
\ref{alg:improved} was run 1000 times using 5, 10, 20 and 50
sampled gene trees per trial, 500 times using 100 gene
trees, 250 times using 200 gene trees, and 100 times using
500 simulated gene trees per trial. The proportion of trials
which yielded the correct unrooted species tree topology is
reported in Table \ref{tab:caterpillar}. Note that
simulations under the 5-taxon caterpillar tree are also
performed by Liu and Yu \cite{Liu_Neighbor} in order to
assess the performance of their `neighbor joining algorithm
for species trees'. For our choices of branch lengths and
sample sizes, the performance of Algorithm
\ref{alg:improved} seems to be roughly equal to the
performance of `neighbor joining for species trees' (Liu and
Yu \cite{Liu_Neighbor}, Figure 2).  \newline

In a second series of simulations, the underlying species
tree was the tree inferred by Algorithm \ref{alg:improved}
for the nine prokaryotes in Section
\ref{subsect:prokaryotes}, see Figure \ref{fig:all}. Again,
for different internal branch lengths and different numbers
of gene trees per trial, 1000 trials (in the case of 5, 10,
20 and 50 gene trees per trial), and 500 resp. 250 resp. 100
trials (in the case of 100 resp. 200 resp. 500 gene trees
per trial) were run, and the proportions of correctly
inferred unrooted species tree topologies are reported in
Table \ref{tab:prokaryotes}.

Two comments are in order: (1) In fact, for each choice
of the parameter $x$ (and fixed number of gene trees per
trial) two simulations were performed. For the first, the
length of the branch leading to the cherry formed by MJ and AF
was set to $4x$, while in a second simulation this branch
length was set to $2x$. The differences in the proportions
of successful trials are small in most cases (between one and two
percent), and, as expected, in most cases the proportion of
successful trials was bigger in the first situation.

(2) The number of gene trees (or rather, the number of
unrooted quartet gene trees for each 4-taxon subset) used
for the reconstruction of the prokaryote species tree in
Section \ref{subsect:prokaryotes} where 21 and 28,
respectively. From Table \ref{tab:prokaryotes} we see that
such a species tree $S$ is likely to be inferred correctly
by Algorithm \ref{alg:improved} if its internal branch
lengths are around $0.5$ (with a probability of more than 90
percent). For branch lengths around $0.2$, however, the
probability for a correct inference of $S$ decreases to
about 40 to 50 percent. (However, there might still be
certain clades on $S$ which can be detected with high
accuracy also for smaller branch lengths.) Clearly, in these
considerations we ignore effects such as horizontal gene
transfer, for whose existence there is evidence in the case
of the nine prokaryotes considered in Section
\ref{subsect:prokaryotes} for some non-ribosomal proteins
(see Teichmann and Mitchison \cite{Teichmann_Mitchison}).

\newpage
\begin{table}
	\centering
\begin{tabular}{ | c | | c | c | c | c | c | c | c |}
	\hline
	$x =$ internal & & & & & & & \\
	branch length & 5 & 10 & 20 & 50 & 100 & 200 & 500 \\
	\hline
	\hline
	0.01 & 0.107 & 0.100& 0.093 & 0.114 & 0.124 & 0.20 & 0.21 \\
	0.05 & 0.162 & 0.210 & 0.230 & 0.330 & 0.430 & 0.67 &
	0.85 \\
	0.1 & 0.285 & 0.334 & 0.440 & 0.662 & 0.815 & 0.93 &
	1.00 \\
	0.2 & 0.428 & 0.580 & 0.751 & 0.943 & 0.990 & 1.00 &
	1.00 \\
	0.5 & 0.755 & 0.890 & 0.980 & 1.000 & 1.000 & & \\
	1.0 & 0.954 & 0.992 & 1.000 & 1.000 & & & \\
	\hline
\end{tabular}
\caption{Simulation results for the 5-taxon caterpillar
tree $(a,(b,(c,(d,e))))$ with all internal branch lengths set
to $x$ (in coalescent units). Table entries are proportions of trials which yielded the correct
unrooted species tree topology. Columns are labelled by
the number of simulated gene trees used in each trial.}
\label{tab:caterpillar}
\end{table}

\begin{table}
	\centering
\begin{tabular}{ | c | | c | c | c | c | c | c | c | c |}
	\hline
	$x =$ int. & & & & & & & & \\
	branch l. & 5 & 10 & 20 & 25 & 50 & 100 & 200 & 500 \\
	\hline
	\hline
	a) 0.1 & 0.005 & 0.027 & 0.068 & 0.102 & 0.230 & 0.320 & 0.33 &
	0.37 \\
	b) 0.1 & 0.007 & 0.018 & 0.064 & 0.087 & 0.218 & 0.282 & 0.31 &
	0.27 \\
	\hline
	a) 0.2 & 0.046 & 0.181 & 0.406 & 0.498 & 0.750 & 0.876 & 0.97
	& 1.00\\
	b) 0.2 & 0.034 & 0.157 & 0.385 & 0.475 & 0.735 & 0.902
	& 0.96 & 1.00\\
	\hline
	a) 0.5 & 0.325 & 0.626 & 0.899 & 0.966 & 0.999 & 1.000 & & \\
	b) 0.5 & 0.303 & 0.607 & 0.911 & 0.957 & 1.000 & 1.000 & & \\
	\hline
	a) 1.0 & 0.703 & 0.873 & 0.966 & 0.980 & 0.999 & 1.000 & & \\
	b) 1.0 & 0.708 & 0.868 & 0.969 & 0.983 & 0.997 & 1.000 & & \\
	\hline
\end{tabular}
\caption{Simulation results for the prokaryote species tree
shown in Figure \ref{fig:all}. All internal branch lengths were
set to $x$ (in coalescent units), except for the edge leading from the root to
(MJ,AF): the length of this edge was set to a) $4x$
and to b) $2x$, respectively. Columns are labelled by the
number of simulated gene trees used in each trial.}
\label{tab:prokaryotes}
\end{table}

\end{document}